%% file: fullmu.tex
\documentclass{llncs}

\usepackage{stmaryrd}
\usepackage{mathtools}
\usepackage{algorithm}
\usepackage{algorithmic}
\usepackage{xspace}
\usepackage{amssymb}
\usepackage{pstricks}
\usepackage{pst-node}
\usepackage{rotating}

\input{macros}

\floatname{algorithm}{Procedure}

\begin{document}

\pagestyle{plain}

\title{A Saturation Method for the Modal Mu-Calculus with Backwards Modalities
      over Pushdown Systems}

\author{{M. Hague} and {C.-H. L. Ong}} 

\institute{Oxford University Computing Laboratory \\
\texttt{Matthew.Hague@comlab.ox.ac.uk} $\quad$ \texttt{Luke.Ong@comlab.ox.ac.uk}
\vspace{-2.5ex} }

\date{\today}

\maketitle

\begin{abstract}
We present an extension of an algorithm for computing \emph{directly} the
denotation of a \mucalc formula $\chi$ over the configuration graph of a
pushdown system to allow backwards modalities.  Our method gives the first
extension of the saturation technique to the full \mucalc with backwards
modalities.  
\end{abstract}

\input{introduction}

\input{notation}

\input{algorithm}

\input{termination}

\input{correctness}

\input{conclusion}

\bibliographystyle{plain}
\bibliography{references}

\end{document}

%% file: macros.tex
\newcommand\dispatch[4]{Dispatch({#1}, {#2}, {#3}, {#4})} 
\newcommand\muand[5]{And({#1}, {#2}, {#3}, {#4}, {#5})}

\newcommand\mubackbox[4]{BackBox({#1}, {#2}, {#3}, {#4})}
\newcommand\mubackdiamond[4]{BackDiamond({#1}, {#2}, {#3}, {#4})}

\newcommand\powerset{\mathcal{P}}

\newcommand\setcomp[2]{\left\{\ {#1}\ \left|\ {#2}\ \right.\right\}} 
\newcommand\set[1]{\left\{{#1}\right\}}

\newcommand\prepop[1]{\overline{Pop}({#1})}
\newcommand\prerew[2]{\overline{Rew}({#1}, {#2})}
\newcommand\prepush[3]{\overline{Push}({#1}, {#2}, {#3})}
\newcommand\pre[3]{Pre({#1}, {#2}, {#3})}

\newcommand{\pds}{\mathbb{P}}
\newcommand{\controlstates}{\mathcal{P}}

\newcommand{\command}[4]{#1\,#2 \rightarrow #4\,#3}
\newcommand{\alphabet}{\Sigma}

\newcommand{\eword}{\varepsilon}
\newcommand{\config}[2]{{\langle #1, #2 \rangle}}
\newcommand{\pdstran}{\hookrightarrow}

\newcommand{\configs}{\mathcal{C}}

\newcommand{\sbot}{\perp}

\newcommand{\ma}{A}
\newcommand{\mastates}{\mathcal{Q}}
\newcommand{\madelta}{\Delta}
\newcommand{\mafinals}{\mathcal{F}}
\newcommand{\mainit}{I}
\newcommand{\matrans}[3]{#1 \xrightarrow{#2} #3}
\newcommand{\matransA}[4]{#1 \xrightarrow[#4]{#2} #3}

\newcommand{\mamonleq}{\preceq}

\newcommand{\runsup}{\ll}

\newcommand{\mucalc}{modal $\mu$-calculus\xspace}

\newcommand{\muenv}{V}
\newcommand{\mupropenv}{\rho}
\newcommand{\setmuvar}{\mathcal{Z}}

\newcommand\mubrack[1]{\llbracket {#1} \rrbracket} 

\newcommand{\backbox}{\overline{\Box}}
\newcommand{\backdiamond}{\overline{\Diamond}}

\newcommand{\genericnode}[6][]{\ifthenelse{\equal{#1}{}}
                                          {#2{#3}{\makebox(#4,#5)[c]{#6}}}
                                          {#2[#1]{#3}{\makebox(#4,#5)[c]{#6}}}}

\newcommand\proofcase[2]{\medskip\noindent\textbf{Case {#1}}: 

\nopagebreak

\medskip\noindent{#2}}

\newcommand\order[1]{\mathcal{O}\left({#1}\right)}

%% file: introduction.tex
\section{Introduction}

Recently we introduced a saturation method for directly computing the denotation
of a \mucalc formula over the configuration graph of a pushdown
system~\cite{HO10b}.  Here we show how this algorithm can be extended to allow
backwards modalities.  This article is intended as a companion to our previous
work, and as such, does not repeat many of the details.

%% file: notation.tex
\section{Preliminaries}

Since we extend our definition of \mucalc, we give the full details here.  The
reader is directed to our previous work for the remaining
preliminaries~\cite{HO10b}.

Given a set of propositions $AP$ and a disjoint set of variables $\setmuvar$,
formulas of the \mucalc are defined as follows (with $x \in AP$ and $Z \in
\setmuvar$):
\[
    \varphi\ :=\ x\ |\ \neg x\ |\ Z\ |\ \varphi\ \land\ \varphi\ |\
    \varphi\ \lor\
    \varphi\ |\ \Box\varphi\ |\ \Diamond\varphi\ |\ \mu Z. \varphi\ |\ \nu
    Z.\varphi \ .
\]
Thus we assume that the formulas are in \emph{positive form}, in the sense that
negation is only applied to atomic propositions. Over a pushdown
system, the semantics of a formula $\varphi$ are given with respect to a
\emph{valuation} $\muenv : \setmuvar \rightarrow \powerset(\configs)$ which maps
each free variable to its set of satisfying configurations and an environment
$\mupropenv : AP \rightarrow \powerset(\configs)$ mapping each atomic
proposition to its set of satisfying configurations.  We then have, 
\[
    \begin{array}{rcl}
        \mubrack{x}^\pds_\muenv &=& \mupropenv(x) \\

        \mubrack{\neg x}^\pds_\muenv &=& \configs \setminus \mupropenv(x) \\

        \mubrack{Z}^\pds_\muenv &=& \muenv(Z) \\

        \mubrack{\varphi_1 \land \varphi_2}^\pds_\muenv &=&
        \mubrack{\varphi_1}^\pds_\muenv \cap \mubrack{\varphi_2}^\pds_\muenv \\

        \mubrack{\varphi_1 \lor \varphi_2}^\pds_\muenv &=&
        \mubrack{\varphi_1}^\pds_\muenv \cup \mubrack{\varphi_2}^\pds_\muenv \\

        \mubrack{\Box \varphi}^\pds_\muenv &=& \setcomp{c \in \configs}{\forall
        c'.c \pdstran c' \Rightarrow c' \in \mubrack{\varphi}^\pds_\muenv} \\

        \mubrack{\Diamond\varphi}^\pds_\muenv &=& \setcomp{c \in \configs}{\exists c'.c
        \pdstran c' \land c' \in \mubrack{\varphi}^\pds_\muenv} \\

        \mubrack{\backbox \varphi}^\pds_\muenv &=& \setcomp{c \in
        \configs}{\forall c'.c' \pdstran c \Rightarrow c' \in
        \mubrack{\varphi}^\pds_\muenv} \\

        \mubrack{\backdiamond\varphi}^\pds_\muenv &=& \setcomp{c \in
        \configs}{\exists c'.c' \pdstran c \land c' \in
        \mubrack{\varphi}^\pds_\muenv} \\

        \mubrack{\mu Z.\varphi}^\pds_\muenv &=& \bigcap\setcomp{S \subseteq
        \configs}{\mubrack{\varphi}^\pds_{\muenv[Z \mapsto S]} \subseteq S} \\

        \mubrack{\nu Z.\varphi}^\pds_\muenv & = & \bigcup\setcomp{S \subseteq
        \configs}{S \subseteq \mubrack{\varphi}^\pds_{\muenv[Z \mapsto S]}} \\
    \end{array}
\]
where $\muenv[Z \mapsto S]$ updates the valuation $\muenv$ to map the variable
$Z$ to the set $S$.

The operators $\Box\varphi$ and $\Diamond\varphi$ assert
that $\varphi$ holds after all possible transitions and after some transition
respectively; $\backbox$ and $\backdiamond$ are their backwards time
counterparts; and the $\mu$ and $\nu$ operators specify greatest and least fixed
points.  Another interpretation of these operators is given below.  For a full
discussion of the \mucalc we refer the reader to a survey by Bradfield and
Stirling~\cite{BS02}.

%% file: algorithm.tex
\section{The Algorithm}

\label{thealg}

\label{algassume}

Without loss of generality, assume all pushdown commands are
$\command{p}{a}{\eword}{p'}$, $\command{p}{a}{b}{p'}$, or
$\command{p}{a}{bb'}{p'}$.

The extensions to our earlier work~\cite{HO10b} are given in
Procedures~\ref{backboxalg} and~\ref{backdiaalg}.  We refer the reader to the
original article for a description of the notations used.  

For a control state $p$ and characters $a, b$, let $\prepop{p} = \setcomp{(p',
a')}{\command{p'}{a'}{\eword}{p}}$, and $\prerew{p}{a} =
\setcomp{(p',a')}{\command{p'}{a'}{b}{p}}$, $\prepush{p}{a}{b} = \setcomp{(p',
a')}{\command{p'}{a'}{ab}{p}}$, and together $\pre{p}{a}{b} = \prepop{p} \cup
\prerew{p}{a} \cup \prepush{p}{a}{b}$.  

\begin{algorithm}
    \caption{\label{backboxalg}$\mubackbox{\ma}{\varphi_1}{c}{\pds}$}

    \begin{algorithmic}
        \STATE {$((\mastates_1, \alphabet, \madelta_1, \_, \mafinals_1),
               \mainit_1) = \dispatch{\ma}{\varphi_1}{c}{\pds}$}

        \STATE {$\ma' = (\mastates_1 \cup \mainit \cup \mastates_{int},
               \alphabet, \madelta_1 \cup \madelta', \_, \mafinals_1)$}

        \STATE {where $\mainit = \setcomp{(p, \backbox \varphi_1, c)}{p \in
               \controlstates}$}

        \STATE {and $\mastates_{int} = \setcomp{(p, \backbox \varphi_1, c, a)}{p
               \in \controlstates \land a \in \alphabet}$}

        \STATE {and $\madelta'=$}

        \STATE {$\begin{array}{l} \setcomp{((p, \backbox \varphi_1, c), a, Q)}{
                           \begin{array}{c}
                               Q = \set{(p, \backbox \varphi_1, c, a)} \cup
                               Q_{pop} \cup Q_{rew}\ \land \\

                               \prepop{p} = \set{(p_1, a_1),\ldots,(p_{n},
                               a_{n})}\ \land \\

                               \bigwedge\limits_{1 \leq j \leq n} \left(
                               \matransA{\mainit_1(p_j)}{a_j}{Q'_j}{\madelta_1}\matransA{}{a}{Q^{pop}_j}{\madelta_1}
                               \right)\ \land \\

                               Q_{pop} = Q^{pop}_1 \cup \cdots \cup Q^{pop}_n\
                               \land \\

                               \prerew{p}{a} = \set{(p'_1,
                               a'_1),\ldots,(p'_{n'}, a'_{n'})}\ \land \\

                               \bigwedge\limits_{1 \leq j \leq n'} \left(
                               \matransA{\mainit_1(p'_j)}{a'_j}{Q^{rew}_j}{\madelta_1}
                               \right)\ \land \\

                               Q_{rew} = Q^{rew}_1 \cup \cdots \cup Q^{rew}_n

                           \end{array}
                       } \cup \\

                       \setcomp{\left((p, \backbox \varphi_1, c, a), b,
                       Q\right)}{   
                           \begin{array}{c}
                               \pre{p}{a}{b} = \set{(p_1, a_1),\ldots,(p_{n},
                               a_{n})} \land \\

                               \bigwedge\limits_{1 \leq j \leq n} \left(
                               \matransA{\mainit_1(p_j)}{a_j}{Q^{push}_j}{\madelta_1}
                               \right) \land \\

                               Q = Q^{push}_1 \cup \cdots \cup Q^{push}_n
                           \end{array}
                       } \cup \\

                       \setcomp{\left((p, \backbox \varphi_1, c), a,
                       \set{q^\ast}\right)}{\forall b . \pre{p}{a}{b} =
                       \emptyset} \cup \\

                       \setcomp{\left((p, \backbox \varphi_1, c), \sbot,
                       \set{q^\eword_f}\right)}{\forall a . \pre{p}{\sbot}{a} =
                       \emptyset} \cup \\

                       \setcomp{\left((p, \backbox \varphi_1, c, a), b,
                       \set{q^\ast}\right)}{\prepush{p}{a}{b} = \emptyset} \cup
                       \\

                       \setcomp{\left((p, \backbox \varphi_1, c, a), \sbot,
                       \set{q^\eword_f}\right)}{\prepush{p}{a}{\sbot} =
                       \emptyset} \\

                       \end{array}$}

        \RETURN {$(\ma', \mainit)$}

    \end{algorithmic}

\end{algorithm}

\begin{algorithm}
    \caption{\label{backdiaalg}$\mubackdiamond{\ma}{\varphi_1}{c}{\pds}$}

    \begin{algorithmic}
        \STATE {$((\mastates_1, \alphabet, \madelta_1, \_, \mafinals_1),
               \mainit_1) = \dispatch{\ma}{\varphi_1}{c}{\pds}$}

        \STATE {$\ma' = (\mastates_1 \cup \mainit \cup \mastates_{int},
               \alphabet, \madelta_1 \cup \madelta', \_, \mafinals_1)$}

        \STATE {where $\mainit = \setcomp{(p, \backdiamond \varphi_1, c)}{p \in
               \controlstates}$}

        \STATE {and $\mastates_{int} = \setcomp{(p, \backbox \varphi_1, c, a)}{p
               \in \controlstates \land a \in \alphabet}$}

        \STATE {and $\madelta' =
                \begin{array}{c} 
                    \setcomp{((p, \backdiamond \varphi_1, c), a, Q)}{
                    \begin{array}{c}
                        (p', a') \in \prepop{p} \land \\

                        \matransA{\mainit_1(p')}{a'}{Q'}{\madelta_1}\matransA{}{a}{Q}{\madelta_1}
                    \end{array}} \cup \\

                    \setcomp{((p, \backdiamond \varphi_1, c), a, Q)}{
                    \begin{array}{c}
                        (p', a') \in \prerew{p}{a} \land \\

                        \matransA{\mainit_1(p')}{a'}{Q}{\madelta_1}
                    \end{array}} \cup \\

                    \set{((p, \backdiamond \varphi_1, c), a, \set{(p,
                    \backdiamond \varphi_1, c, a)})} \cup \\

                    \setcomp{((p, \backdiamond \varphi_1, c, a), b, Q)}{
                    \begin{array}{c}
                        (p', a') \in \prepush{p}{a}{b} \land \\

                        \matransA{\mainit_1(p')}{a'}{Q}{\madelta_1}
                    \end{array}}
                \end{array}$.}

        \RETURN {$(\ma', \mainit)$}

    \end{algorithmic}

\end{algorithm}

%% file: termination.tex
\section{Termination}

\label{termination}

The new procedures defined here add extra cases to the termination
proof~\cite{HO10b}.  We show these cases here and refer the reader to the
original article for an explanation of the notation and concepts.

\begin{samepage}
    \begin{lemma}[Termination]
    \label{terminates}
        The algorithm satisfies the following properties.
        \begin{enumerate}
        \item Each subroutine introduces a fixed set of new states, independent of
              the automaton $\ma$ given as input (but may depend on the other
              parameters).  Transitions are only added to these new states.

        \item For two input automata $\ma_1$ and $\ma_2$ (giving valuations of the same
              environments) such that $\ma_1 \mamonleq \ma_2$, then the returned
              automata $\ma'_1$ and $\ma'_2$, respectively, satisfy $\ma'_1 \mamonleq
              \ma'_2$.

        \item The algorithm terminates.
        \end{enumerate}
    \end{lemma}
\end{samepage}
\begin{proof}
    The first of these conditions is trivially satisfied by all constructions,
    hence we omit the proofs.  Similarly, termination is trivial.  The second
    and third conditions will be shown by mutual induction over the recursion
    (structure of the formula).  The new cases follow.

    \proofcase{$\mubackbox{\ma}{\varphi_1}{c}{\pds}$ and
    $\mubackdiamond{\ma}{\varphi_1}{c}{\pds}$}{%
        It can be observed that all new transitions in $\ma$ are derived from
        transitions $\matransA{\mainit(p')}{a}{Q}{\ma}$ (or are independent of
        $\ma$ and $\ma'$).  Since $\ma \mamonleq \ma'$ it follows that all
        transitions have a counterpart $\matransA{\mainit(p')}{a}{Q'}{\ma'}$
        with $Q' \runsup Q$.  Hence the property follows in a similar manner to
        the previous cases. } 
\end{proof}

\subsection{Complexity}

The new procedures change the complexity of the algorithm slightly, although the
algorithm remains in EXPTIME.  In particular, the algorithm is now exponential
in the number of control states, the size of the stack alphabet and the size of
the formula.  Let $m$ be the nesting depth of the fixed points of the formula
and $n$ be the number of states in $\ma_\muenv$.  We introduce at most $k =
\order{|\controlstates|\cdot|\chi|\cdot m\cdot |\alphabet|}$ states to the
automaton.  Hence, there are at most $\order{n+k}$ states in the automaton
during any stage of the algorithm.  The fixed point computations iterate up to
an $\order{2^{\order{n + k}}}$ number of times.  Each iteration has a recursive
call, which takes up to $\order{2^{\order{n + k}}}$ time.  Hence the algorithm
is $\order{2^{\order{n + k}}}$ overall.

%% file: correctness.tex
\section{Correctness}
\label{correctness}

We extend the proofs of correctness.  We refer the reader to our previous work
for the full details~\cite{HO10b}.

\begin{samepage}
    \begin{definition}[Correctness Conditions]
        The correctness conditions are as follows.  Let $\ma$ be the input
        automaton, $\varphi$ be the input formula\footnote{For cases such as
        $\muand{\ma}{\varphi_1}{\varphi_2}{c}{\pds}$ we take, as appropriate
        $\varphi = \varphi_1 \land \varphi_2$.}, $c$ be the input level and $\ma'$
        be the result.
        \begin{enumerate}
        \item We only introduce level $c$ states.

        \item If $\ma$ is $\muenv$-sound, $\ma'$ is $\muenv^c_\varphi$-sound.

        \item If $\ma$ is $\muenv$-complete, $\ma'$ is $\muenv^c_\varphi$-complete.
        \end{enumerate}
    \end{definition}
\end{samepage}

The first condition is obvious.  The remaining conditions are shown by induction
and require the addition of proof cases for the new procedures.

\begin{lemma}[Valuation Soundness]
\label{issound}
    The algorithm is $\muenv$-sound.
\end{lemma}
\begin{proof}
    \proofcase{$\mubackbox{\ma}{\varphi_1}{c}{\pds}$}{%    
        We assume that $\ma$ is valuation sound with respect to some valuation
        $\muenv$.  By induction the result $\ma_1$ of the recursive call is
        valuation sound with respect to $\muenv^c_{\varphi_1}$.  We show that
        $\ma'$ is valuation sound with respect to
        $\muenv^c_{\backbox\varphi_1}$.

        We observe that no $(p', \backbox \varphi_1, c)$ are reachable from a
        state $(p, \backbox \varphi, c, a)$, hence we show soundness for the
        latter states first.

        The first case is for some $b$ with $\prepush{p}{a}{b} = \emptyset$.  In
        this case, the valuation of $(p, \backbox \varphi, c, a)$ contains all
        words of the form $bw$.  Hence soundness is immediately satisfied.

        Otherwise, $\prepush{p}{a}{b} = \set{(p_1, a_1),\ldots,(p_n, a_n)}$ such
        that for all $1 \leq j \leq n$, $\config{p_j}{a_iw} \pdstran
        \config{p}{abw}$.  Take a new transition $((p, \backbox \varphi_1, c,
        a), b, Q)$ derived from the runs
        $\matransA{\mainit_1(p_j)}{a_j}{Q_j}{\ma_1}$ for all $1 \leq j \leq n$,
        with $Q = Q_1 \cup Q_n$.  Suppose for some $w$, $w \in
        \muenv^c_{\backbox\varphi_1}(q)$ for all $q \in Q$.  By valuation
        soundness of $\ma_1$ we know $a_jw \in
        \muenv^c_{\backbox\varphi_1}(\mainit_1(p_j))$ and hence, since all
        transitions to $\config{p}{abw}$ are from configurations satisfying
        $\varphi_1$, we have $bw \in \muenv^c_{\backbox\varphi_1}(p, \backbox
        \varphi_1, c, a)$ as required.

        The remaining states are of the form $(p, \backbox \varphi_1, c)$.  We
        first deal with the case when for all $b$ we have $\pre{p}{a}{b} =
        \emptyset$.  In this case, the valuation of $\backbox \varphi_1$
        contains all words of the form $aw$ for some $w$.  Hence, all added
        transitions are trivially sound.

        Otherwise, take a new transition $((p, \backbox \varphi_1, c), a, Q)$
        derived from some $b$, the value of $\prepop{p} = \set{(p_1,
        a_1),\ldots,(p_n, a_n)}$ and for all $1 \leq j \leq n$, the runs
        $\matransA{\mainit_1(p_j)}{w_j}{Q'_j}{\ma_1}\matransA{}{b}{Q^{pop}_j}{\ma_1}$,
        with $Q_{pop} = Q^{pop}_1 \cup Q^{pop}_n$, and the value of $\prerew{p}
        = \set{(p'_1, a'_1),\ldots,(p'_{n'}, a'_{n'})}$ and for all $1 \leq j
        \leq n'$, the runs $\matransA{\mainit_1(p'_j)}{a'_j}{Q^{rew}_j}{\ma_1}$,
        with $Q_{rew} = Q^{rew}_1 \cup Q^{rew}_n$.  Finally, $Q = \set{(p,
        \backbox \varphi_1, c, a, b)} \cup Q_{pop} \cup Q_{rew}$. 
        
        Suppose for some $w$, $w \in \muenv^c_{\backbox \varphi_1}(q)$ for all
        $q \in Q_{pop}$.  By valuation soundness of $\ma_1$ we know $a_jaw \in
        \muenv^c_{\backbox\varphi_1}(\mainit_1(p_j))$ and hence all pop
        transitions leading to $\config{p}{aw}$ are from configurations
        satisfying $\varphi_1$. 
        
        Now suppose for some $aw$, $aw \in \muenv^c_{\backbox \varphi_1}(q)$ for
        all $q \in Q_{rew}$.  By valuation soundness of $\ma_1$ we know $a_jw
        \in \muenv^c_{\backbox\varphi_1}(\mainit_1(p_j))$ and hence all rewrite
        transitions leading to $\config{p}{aw}$ are from configurations
        satisfying $\varphi_1$.

        Finally, consider some $bw$ in the valuation of $(p, \backbox \varphi_1,
        c, a)$.  From the soundness of this state, shown above, we have that all
        push transitions leading to $\config{p}{abw}$ are from configurations
        satisfying $\varphi_1$.

        Putting the three cases together, we have for all $abw \in
        \muenv^c_{\backbox\varphi_1}(p, \backbox \varphi_1, c)$ as required. 
        
        The above cases do not cover the case $\sbot \in \muenv^c_{\backbox
        \varphi_1}(p, \backbox \varphi_1, c)$.  However, since no push
        transition can reach this stack, we just require the first two cases and
        that $(p, \backbox \varphi_1, c, \sbot) = q^\eword_f$. }

    \proofcase{$\mubackdiamond{\ma}{\varphi_1}{c}{\pds}$}{% 
        We assume that $\ma$ is valuation sound with respect to some valuation
        $\muenv$.  By induction the result $\ma_1$ of the recursive call is
        valuation sound with respect to $\muenv^c_{\varphi_1}$.  We show that
        $\ma'$ is valuation sound with respect to $\muenv^c_{\backdiamond
        \varphi_1}$.

        We begin with the states $(p, \backdiamond, c, a)$.  Take a transition
        $((p, \backdiamond, c, a), b, Q)$.  Then there is some $(p', a') \in
        \prepush{p}{a}{b}$ such that $\matrans{\mainit_1(p')}{a'}{Q}{\ma_1}$.
        From the soundness of $\ma_1$ we know for all $w$ with $w \in
        \muenv^c_{\backdiamond \varphi_1}(q)$ for all $q \in Q$ we have $a'w \in
        \muenv^c{\backdiamond \varphi_1}(\mainit_1(p'))$.  Since
        $\config{p'}{a'w} \pdstran \config{p}{abw}$ we have $\config{p}{abw}$
        satisfies $\varphi_1$ and hence $bw \in \muenv^c_{\backdiamond
        \varphi_1}(p, \backdiamond, c, a)$ and the transition is sound.

        For the remaining states, take a new transition $((p, \backdiamond
        \varphi_1, c), a, Q)$.  There are three cases.
        
        If the transition was derived from some $(p', a') \in \prepop{p}$ and
        the run $\matransA{\mainit_1(p')}{a'a}{Q}{\ma_1}$, then suppose for some
        $w$, $w \in \muenv^c_{\backdiamond \varphi_1}(q)$ for all $q \in Q$.  By
        valuation soundness of $\ma_1$ we know $a'aw \in
        \muenv^c_{\backdiamond\varphi_1}(\mainit_1(p'))$ and hence, since there
        is a transition $\config{p'}{a'aw}$, a configuration satisfying
        $\varphi_1$, to $\config{p}{aw}$ we obtain $aw \in
        \muenv^c_{\backdiamond\varphi_1}(p, \backdiamond \varphi_1, c)$ as
        required.  
        
        If the transition was derived from some $(p', a') \in \prerew{p}{a}$ and
        the run $\matransA{\mainit_1(p')}{a'}{Q}{\ma_1}$, then suppose for some
        $w$, $w \in \muenv^c_{\backdiamond \varphi_1}(q)$ for all $q \in Q$.  By
        valuation soundness of $\ma_1$ we know $a'w \in
        \muenv^c_{\backdiamond\varphi_1}(\mainit_1(p'))$ and hence, since there
        is a transition $\config{p'}{a'w}$, a configuration satisfying
        $\varphi_1$, to $\config{p}{aw}$ we obtain $aw \in
        \muenv^c_{\backdiamond\varphi_1}(p, \backdiamond \varphi_1, c)$ as
        required. 
       
        Finally, if $Q = \set{(p, \backdiamond, c, a)}$ then soundness is
        immediate from the definition of $\muenv^c_{\backdiamond \varphi_1}$.  }
\end{proof}

\begin{lemma}[Valuation Completeness]
\label{iscomplete}
    The algorithm is $\muenv$-complete.
\end{lemma}
\begin{proof}
    \proofcase{$\mubackbox{\ma}{\varphi_1}{c}{\pds}$}{% 
        We are given that $\ma$ is valuation complete with respect to some
        valuation $\muenv$, and by induction we have completeness of the result
        $\ma_1$ of the recursive call with respect to $\muenv^c_{\varphi_1}$.
        We show $\ma'$ is complete with respect to
        $\muenv^c_{\backbox\varphi_1}$.

        As in the soundness proof, we begin with the states $(p, \backbox
        \varphi_1, c, a)$.  In the case $\prepush{p}{a}{b} = \emptyset$ for some
        $b$, we either have $b = \sbot$ and the transition from $(p, \Box
        \varphi_1, c, a)$ to $\set{q^\eword_f}$ witnesses completeness, or we
        have $a \neq \sbot$ and the transition to $\set{q^\ast}$ witnesses
        completeness.

        Otherwise $\prepush{p}{a}{b} = \set{(p_1, a_1),\ldots,(p_n, a_n)}$.
        Take some $bw$ such that $abw \in \muenv^c_{\backbox \varphi_1}(p,
        \backbox \varphi_1, c, a)$.  Then we have $a_jw \in \muenv^c_{\backbox
        \varphi_1}(p_j, \varphi_1, c)$ for all $1 \leq j \leq n$.  From
        completeness of $\ma_1$ we have a transition
        $\matrans{\mainit_1(p_j)}{a_j}{Q_j}$ with $w \in \muenv^c_{\backbox
        \varphi_1}(q)$ for all $q \in Q_j$.  Hence, we have a complete
        $b$-transition from $(p, \backbox \varphi_1, c, a)$ as required.

        For the states of the form $(p, \backbox \varphi_1, c)$ we first deal
        with the case when for all $b$ we have $\pre{p}{a}{b} = \emptyset$.  In
        this case we immediately have transitions witnessing completeness.
        
        Otherwise, take some $abw \in \muenv^c_{\backbox \varphi_1}(p, \backbox
        \varphi_1, c)$.  Then, for all $(p', a') \in \prepop{p}$,  we have
        $a'abw \in \muenv^c_{\backbox \varphi_1}(\mainit_1(p'))$; and for all
        $(p', a') \in \prerew{p}{a}$ we have $a'bw \in \muenv^c_{\backbox
        \varphi_1}(\mainit_1(p'))$; and for all $(p', a') \in \prepush{p}{a}{b}$
        we have $a'w \in \muenv^c_{\backbox \varphi_1}{\mainit_1(p')}$.  From
        completeness of $\ma_1$ we have a complete run
        $\matransA{\mainit_1(p')}{a'}{Q'}{\ma_1}\matransA{}{a}{Q}{\ma_1}$ for
        each $(p', a') \in \prepop{p}$ and a complete run
        $\matransA{\mainit_1(p')}{a'}{Q}{\ma_1}$ for each $(p', a') \in
        \prerew{p}{a}$.  Since we know $bw \in \muenv^c_{\backbox \varphi_1}(p,
        \backbox \varphi_1, c, a)$ there must be some complete transition from
        $(p, \backbox \varphi_1, c)$ as required.

        The only case not covered by the above is the case $\sbot \in
        \muenv^c_{\backbox \varphi_1}(p, \backbox, \varphi_1, c)$.  In this case
        there are no push transitions reaching this configuration.  That is
        $\prepush{p}{\sbot}{b} = \emptyset$ for all $b$.  Note also that we
        equated all $(p, \backbox \varphi_1, c, \sbot)$ with $q^\eword_f$.
        Hence, from the pop and rewrite cases above, and that $(p, \backbox
        \varphi_1, c, \sbot) = q^\eword_f$ we have completeness as required.  }

     \proofcase{$\mubackdiamond{\ma}{\varphi_1}{c}{\pds}$}{% 
         We are given that $\ma$ is valuation complete with respect to some
         valuation $\muenv$, and by induction we have completeness of the result
         $\ma_1$ of the recursive call with respect to $\muenv^c_{\varphi_1}$.
         We show $\ma'$ is complete with respect to
         $\muenv^c_{\backdiamond\varphi_1}$.  There are three cases.

         Assume some $aw$ such that $aw \in \muenv^c_{\backdiamond\varphi_1}(p,
         \backdiamond \varphi_1, c)$ by virtue of some $(p', a') \in \prepop{p}$
         such that we have $\config{p'}{a'aw} \in
         \muenv^c_{\backdiamond\varphi_1}(\mainit_1(p'))$.  By completeness of
         $\ma_1$ we have a run $\matransA{\mainit_1(p')}{a'a}{Q}{\ma_1}$ such
         hat for all $q \in Q$, $w \in \muenv^c_{\backdiamond\varphi_1}(q)$.
         Hence, the transition $((p, \backdiamond\varphi_1, c), a, Q)$ witnesses
         completeness.  

         Otherwise, take some $aw$ such that $aw \in
         \muenv^c_{\backdiamond\varphi_1}(p, \backdiamond\varphi_1, c)$ from
         some $(p', a') \in \prerew{p}{a}$ such that we have $\config{p'}{a'w}
         \in \muenv^c_{\backdiamond\varphi_1}(\mainit_1(p'))$.  By completeness
         of $\ma_1$ we have a run $\matransA{\mainit_1(p')}{a'}{Q}{\ma_1}$ such
         that for all $q \in Q$, $w \in \muenv^c_{\backdiamond\varphi_1}(q)$.
         Hence, the transition $((p, \backdiamond\varphi_1, c), a, Q)$ witnesses
         completeness.  

         Finally, take some $abw$ such that $abw \in
         \muenv^c_{\backdiamond\varphi_1}(p, \backdiamond\varphi_1, c)$ from
         some $(p', a') \in \prepush{p}{a}{b}$ such that we have
         $\config{p'}{a'w} \in \muenv^c_{\backdiamond\varphi_1}(\mainit_1(p'))$.
         By completeness of $\ma_1$ we have a run
         $\matransA{\mainit_1(p')}{a'}{Q}{\ma_1}$ such that for all $q \in Q$,
         $w \in \muenv^c_{\backdiamond\varphi_1}(q)$.  Hence, the transitions
         $((p, \backdiamond \varphi_1, c), a, \set{(p, \backdiamond, c, a)})$
         and $((p, \backdiamond\varphi_1, c, a), a, Q)$ witness completeness.  }
\end{proof}

%% file: conclusion.tex
\section{Conclusion and Future Work}

In previous work, we have introduced a saturation method for directly computing
the denotation of a \mucalc formula over the configuration graph of a pushdown
system.  Here, we have shown how to extend this work to allow backwards
modalities.

%% file: fullmu.bbl
\begin{thebibliography}{1}

\bibitem{BS02}
J.~C. Bradfield and C.~P. Stirling.
\newblock Modal logics and mu-calculi: An introduction.
\newblock In {\em Handbook of Process Algebra}, pages 293--330, 2001.

\bibitem{HO10b}
M.~Hague and C.-H.~L. Ong.
\newblock A saturation method for the modal mu-calculus over pushdown systems,
  2010.
\newblock To appear in Information and Computation.

\end{thebibliography}
